\documentclass{wiley-article}

\DeclareMathOperator{\argmin}{argmin}
\usepackage[utf8]{inputenc}
\usepackage{acronym}
\usepackage{graphicx}
\usepackage{subcaption}
\usepackage[normalem]{ulem}
\usepackage{color}
\usepackage{url}
\usepackage{mathtools}

\usepackage{pbalance}
\usepackage{pdfpages}
\usepackage{xcolor,amsmath}
\usepackage[linesnumbered,ruled,vlined]{algorithm2e}
\RestyleAlgo{ruled}
\SetKwComment{Comment}{\color{green!50!black}// }{}

\SetKwProg{Function}{function}{}{}

\corremail{sebastien.martin@huawei.com}
\corraddress{Huawei Technologies, 18 quai du point du jour, Boulogne-Billancourt, France}

\title{Atomic Column Generation For Consensus Between Algorithms: Application to Path Computation} 
\author{S\'ebastien Martin} 
\author{Pierre Bauguion}
\author{Youcef Magnouche}
\author{J\'er\'emie Leguay}
\affil{Huawei Technologies, 18 quai du point du jour, Boulogne-Billancourt, France}

\begin{document}
\begin{frontmatter}
\maketitle
\begin{abstract}
In real-life applications, most optimization problems are variants of well-known combinatorial optimization problems, including additional constraints to fit with a particular use case. Usually, efficient algorithms to handle a restricted subset of these additional constraints already exist, or can be easily derived, but combining them together is difficult.  
The goal of our paper is to provide a framework that allows merging several so-called atomic algorithms to solve an optimization problem including all associated additional constraints together. The core proposal, referred to as Atomic Column Generation (ACG) and derived from Dantzig-Wolfe decomposition, allows converging to an optimal global solution  with any kind of atomic algorithms. We show that this decomposition improves the continuous relaxation and describe the associated Branch-and-Price algorithm. We consider a specific use case in telecommunication networks where several  Path Computation Elements (PCE) are combined as atomic algorithms to route traffic. We demonstrate the efficiency of ACG on the resource-constrained shortest path problem associated with each PCE and show that it remains competitive with benchmark algorithms.
\keywords{Combinatorial optimization, Column Generation, Dantzig-Wolfe decomposition, Black-box optimization, Path Computation.}
\end{abstract}
\end{frontmatter}

\section{Introduction}

Combinatorial optimization problems involve finding an optimal solution from a finite set of feasible solutions, described by a set of constraints. These constraints can be distinguished into two types: 1) \emph{structural} constraints defining the structure of the solution, such as a path, a tree, an independent set or an assignment, and 2) \emph{additional} constraints imposing new requirements on the solution, such as a bound on the number of edges or nodes, inclusion or exclusion of nodes, conflicts, etc. Often, introducing additional constraints makes the problem much harder to solve, requiring sophisticated and tailored approaches to find the optimal solution.   In most cases, existing algorithms for the problem, never consider all desired additional constraints and they can be hard to adapt for a particular variant. 

In this paper, we develop a novel framework, called Atomic Column Generation ($ACG$), that allows to easily combine efficient existing algorithms, referred to as \emph{atomic}, without the need to develop a dedicated algorithm for the whole problem.   
The main advantage of $ACG$ is to decompose the problem into independent sub-problems (pricing problems) that can be solved using atomic algorithms, all working on partial versions of the same problem. They all produce solutions that follow structural constraints, but each, handles a different sub-set of additional constraints.
As an example, let us consider the spanning tree problem \cite{gower1969minimum} that consists in finding 
a tree with the minimum total cost that spans all nodes. The authors in \cite{kritikos2017greedy} propose an efficient algorithm when the number of nodes in every subtree, connected to the root node, must be lower than a given bound. Moreover, the authors in \cite{bui2006ant} propose an efficient algorithm for a variant where the degree of each vertex is lower than a given bound. Thanks to $ACG$, the two algorithms proposed in \cite{bui2006ant,kritikos2017greedy} can be easily combined to solve the variant of the spanning tree problem with both capacity constraints on the subtrees and node-degree constraints.

Our contribution is strongly related to multi-agent optimization, where several decision-makers try to converge to a consensus solution.
In various methods~\cite{jeong2023review},  decision variables are duplicated to decompose the problem, one for each agent, and agents coordinate through coupling constraints.
To design distributed algorithms with consensus, the Alternating Direction Method of Multipliers (ADMM)~\cite{yang2022survey} can be used. However, it decomposes the problem into a series of nonlinear sub-problems \cite{yao2019admm}, which is not always suitable. Also, in practice, its convergence is highly sensitive to a penalty parameter (stepsize) which is difficult to tune, even if some works have tried to mitigate this issue~\cite{he20121, xu2017adaptive}. In the rest of the paper, we stick to traditional linear programming methods, which quickly and reliably converge for routing problems.
To the best of our knowledge, there is no generic method that allows combining algorithms to solve partial versions of the same problem to find a consensus solution. 

The closest methods are decomposition techniques, such as Dantzig-Wolfe and Benders \cite{GRISET20221067, vanderbeck2006generic}, that offer significant advantages in combinatorial optimization. By solving, iteratively, the master problem and the sub-problems, they can handle complex problems and improve computational efficiency. 
Dantzig-Wolfe decomposition is often used for problems such as generalized constraint shortest path \cite{aneja1978constrained}, multi-commodity flow~\cite{ahuja1995applications} or multi-paths routing~\cite{liu2015multi} where the Column Generation algorithm~\cite{desaulniers2006column} is commonly used.
However, in Dantzig-Wolfe~\cite{vanderbeck2006generic} decomposition, structural constraints are usually kept in the master while only additional constraints are pushed to the pricing problems. 
Even if $ACG$ is based on the Dantzig-Wolfe decomposition, each pricing problem is composed of all structural constraints and one or multiple additional constraints. Hence, sub-problems have the same solution structure as the global problem. While in classical Dantzig-Wolfe decomposition, the master problem selects a collection (not necessarily one) of columns in the final solution, in ACG, the master problem builds a consensus over a set of pricing problems, each solved by an atomic algorithm. In addition, ACG is able to handle black-box solvers as atomic algorithms, specifically designed for a hard additional constraint, or to deal with confidential data.

In this paper, we consider a real telecommunication network use case, called Augmented PCE (Path computation Element), in which merging existing atomic algorithms is significantly important. This use case consists in combining several shortest path algorithms, each handling specific additional constraints like resource constraints or node inclusion. Based on this use case, we introduce and study an example where ACG applies for shortest-path problems.

\textbf{Main contributions.} 

\textit{i) ACG}: we propose a generic  decomposition method based on a Dantzig-Wolfe decomposition over a consensus-based model where some constraints and variables are duplicated to establish a consensus between sub-problems. Its goal is to merge partial solutions of the original problem obtained from a pool of atomic algorithms. 
We show that this decomposition improves the continuous relaxation of the original model and we introduce a Branch-and-Price algorithm to solve it optimally.

\textit{ii) Augmented PCE}: We apply ACG on Generalized Constrained Shortest Path (GCSP) problems which are typically solved by PCE controllers in telecommunication networks to route traffic. In this context, we propose an Augmented PCE architecture and the mathematical decomposition at its core to assemble features from a single or several PCEs.
Indeed, available atomic algorithms inside PCEs can be complex black-boxes to satisfy Quality of Experience (QoE)~\cite{calvigioni2018quality} or deterministic latency~\cite{NCtoolbox} constraints. They could also be dedicated algorithms to address hard constraints that are difficult to combine with others.
The only expectation of the Augmented PCE is that atomic algorithms must return an elementary path of minimum arc cost and be able to verify feasibility for a given path, which is the case in practice. We show that ACG has the capability to merge partial solutions (easier) of the original problem, obtained from atomic algorithms, to find the optimal solution.

\textit{iii) Evaluation:} we study a typical use case where a Resource Constraint Shortest Path (RCSP) problem with explicit upper and lower bound constraints must be solved and we propose, as a benchmark, a new custom algorithm called MultiPulse  (based on Pulse~\cite{LOZANO2013378}). To compare with other methods, we consider RCSP problem as efficient specific algorithms can be derived for it. We evaluate the performance of ACG, against MultiPulse and a Compact Model solved with IBM ILOG CPLEX~\cite{cplx}. We show that it remains competitive against a dedicated algorithm while providing the same level of flexibility as existing programmable TE solvers.

\textbf{Paper structure.} 
A telecommunication use case, called Augmented PCE, is given in Sec.~\ref{sec:system} to motivate why the proposed solution is needed. To help the reader, an application of our method in the Augmented PCE use case is given in Sec.~\ref{sec:optim}. 
Sec.~\ref{sec:general} generalizes the approach to decompose integer programs. Sec.~\ref{sec:results} presents results over RCSP problems along with an efficient benchmark algorithm. Sec.~\ref{sec:conclusion} concludes the paper.

\section{Use case from Telecommunication: the Augmented PCE}
\label{sec:system}
In telecommunication networks, traffic is typically \textit{engineered} to meet a wide range of  requirements~\cite{mendiola2016survey} by a (logically) centralized control entity that is responsible for path computation and network optimization, called a Path Computation Element (PCE)~\cite{paolucci2013survey}. 
When a path computation request is received with a set of constraints (e.g. delay, disjointness) to satisfy, the PCE computes a path and responds with the appropriate routing information.
PCE controllers have to handle a wide range of path computation requirements with various metrics and constraints. While
the most popular ones have been documented~\cite{IANA},
each operator can have its own operational constraints.
Small operators typically use off-the-shelf PCEs, like any other network management tools, whereas large service providers have internal resources to develop their own with a high degree of customization.
In both cases, the main challenge for network management teams is to quickly handle emerging routing requirements when the feature is not available. Indeed, a new routing requirement may call for the development of a dedicated algorithm. Its design and implementation, as well as its integration into the overall PCE, can take a critical amount of time from a business perspective.

In practice, PCE controllers are architected as a set of nested algorithms~\cite{slicing} that forms a hierarchy going from the most basic ones to the most specialized ones. The most basic algorithm is typically an efficient bidirectional Dijkstra~\cite{luby1989bidirectional}  which is then plugged into other algorithms to calculate paths with various constraints. In turn, these algorithms can be plugged inside multi-paths algorithms to find k-shortest paths or a maximally disjoint pair of paths, for instance. 
This approach allows adding on top of existing algorithms new requirements. To satisfy diverse routing requirements, PCE architectures can rapidly become complicated and difficult to maintain because of the dependencies between algorithms. Indeed, the insertion of an algorithm may require the modification of many others. 
To overcome maintenance issues with rigid PCE architectures, generic and programmable traffic engineering solvers like DEFO (Declarative and Expressive Forwarding Optimizer)~\cite{hartert2015declarative} have been proposed. DEFO uses a Domain Specific
Language (DSL) to model the operator’s intents and constraints. The DSL is then translated to a CP (Constraint Programming) or an ILP (Integer Linear Programming) model and pushed to an off-the-shelf solver (e.g. IBM ILOG CPLEX~\cite{cplx}). Using
this solution, new constraints can be quickly developed and automatically combined with existing ones. Unfortunately, such an approach suffers from a number of limitations. First, the use of a generic solver does not scale for real-life applications. Second, developments are required to update the model using a specific paradigm (CP or ILP) and the DSL needs to be modified accordingly.
Finally, operators cannot plug and enjoy the existing algorithm libraries they developed, or trust, to handle part of their routing requirements.

To best trade-off between programmability and re-usability, we propose the \textit{Augmented PCE} architecture. It lets operators quickly extend their own PCE or third-party libraries with new features, without the need to modify them. As shown in the example of Fig.~\ref{fig:features}, different features, i.e. sub-set of \emph{additional} constraints, can be supported by the set of available PCEs, all capable of finding routing paths, i.e. a solution constrained by the same \emph{structural} constraints. Some features can be redundant, with the possibility that one implementation is the best among the others. 
Therefore, the goal of the Augmented PCE is to map each requested constraint with an atomic algorithm. It decides which atomic algorithm to use and with what feature or configuration. User preferences regarding atomic algorithms can be considered when deciding the strategy. The orchestration can be realized solving a matching problem. However, this part is out of the scope of this paper. 

The Augmented PCE provides a way to support requirements that are not supported by any of the available PCEs, or to address them better, by providing faster or better solutions. It can be used by an operator to combine features from different vendors, or to augment its own PCE with external features. 
Note that an atomic algorithm can be called several times. For example, if the original request asks for multiple end-to-end QoS constraints and the only algorithm available solves a constrained shortest path for one metric, multiple instances of this atomic algorithm can be combined to find a solution. The features used by each PCE, called by the Augmented PCE, are not necessarily disjoint. The same feature can be considered by different PCE, which can help the consensus.

\begin{figure}[t]
     \centering    \includegraphics[width=0.8\linewidth]{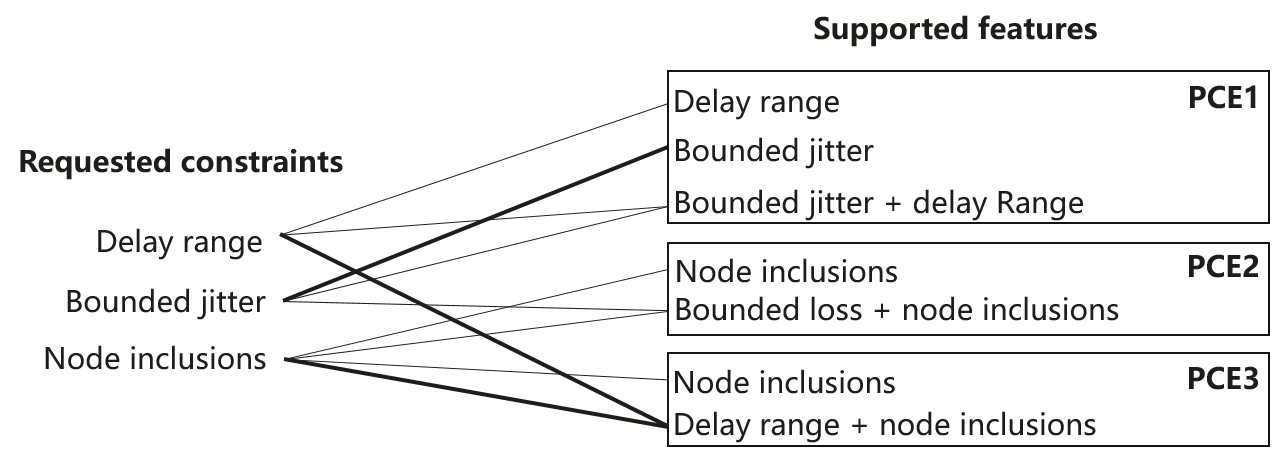}
    \caption{Mapping example between required constraints vs. supported features. Bold lines represent a possible selection of atomic algorithms, i.e. features from PCE1 and PCE3.
    This example shows that the Augmented PCE selects the most appropriate PCEs as atomic algorithms to handle the requested constraints.
    }
    \label{fig:features}
\end{figure}

Several solutions have been proposed for the combination of multiple PCEs. In the PCEP~\cite{rfc4655} protocol, multiple PCEs can be used 
to compute sub-paths and concatenate them. It allows distinct PCEs to be responsible for path computation in their respective domain of authority. 
However, in all cases, only one PCE per domain is selected per request. If constraint A is only available in PCE 1, and B is only available in PCE 2, there is no way to find a path satisfying constraints A and B.

Let us review a non-exhaustive list of path computation problems, solved through atomic algorithms, inside PCEs.
\paragraph{Constrained Shortest Path (CSP)} It consists in finding the shortest path under maximum end-to-end constraints on additive metrics (IGP cost, TE metric, QoS). This problem is weakly NP-hard with one  constraint  and becomes strongly NP-hard with a polynomial number \cite{KORKMAZ2002225}. 
LARAC~\cite{juttner2001lagrange} and GEN-LARAC~\cite{xiao2005gen} are classical heuristic options based on Lagrangian relaxation to satisfy one or multiple  \textit{upper bound} constraints, respectively 
As a sub-routine in these algorithms, an efficient bidirectional Dijkstra~\cite{luby1989bidirectional} is typically used. Some efficient exact dynamic programming algorithms such as Pulse~\cite{LOZANO2013378} have also been proposed for multiple maximum end-to-end constraints. 
\paragraph{Resource Constrained Shortest Path (RCSP)} For this variant, multiple additive end-to-end metrics must have maximum and minimum values. These constraints are called \textit{range} constraints. This problem is strongly NP-hard even for one range \cite{pugliese2013survey}. Range constraints are useful, for instance, to compute a backup path with a QoS close to that of the primary path.
It can be solved using Lagrangian relaxation in combination with dynamic programming~\cite{beasley1989algorithm, pugliese2013survey}. The proposed solutions are heuristics or with an exponential number of states at each node.
\paragraph{Generalized Constrained Shortest Path (GCSP)} This class of problems covers the generic case where constraints can be linear, non-linear, with inclusions, etc, which are strongly NP-hard and even not in NP for some variants. Note that the RCSP is a particular case of the GCSP. Let us take two examples. Inclusion constraints can be handled using the Constrained Shortest Path Tour algorithm~\cite{ferone2016constrained, martin2022constrained}. 
In telecommunication networks with Segment Routing (SR), shortest paths with a bounded number of labels~\cite{jadin2019cg4sr} must be calculated.
GCSP problems have applications in various fields, beyond networking, such as crew scheduling, and crew rostering problems \cite{LOZANO2013378, SHI201713}.

In the rest of this paper, we will step out from the telecommunication application and present how ACG can be used for GCSP problems and generalized for the decomposition of problems where independent sub-problems can be solved using atomic algorithms.

\section{Application of Atomic Column Generation (ACG) }
\label{sec:optim}

We first propose an extended model with an exponential number of variables to establish a consensus among a set of atomic algorithms. It comes from the Dantzig-Wolfe decomposition applied to a consensus-based reformulation of the original model for GCSP. The detailed reformulation, decomposition, and generalization will be given in Sec~\ref{sec:general}. 
Then, we present a practical Branch-and-Price algorithm, called ACG, to derive optimal solutions.

\subsection{Problem Decomposition}

Our goal is to design a Dantzig-Wolfe decomposition for GCSP problems such that pricing problems associated with different requirements can be solved with a set of available atomic algorithms. The mathematical expression of constraints handled by each atomic algorithm can be totally unknown to ACG. 
It is however assumed that the solution returned by atomic algorithms is an elementary path, i.e. a loop-free path, of minimum arc cost from the source to the destination, where the costs can be given as input. Note that this property is valid for all algorithms in practice.
Atomic algorithms must also be able to evaluate the feasibility of a solution with regard to the constraints they handle. It can be simply realized by asking them to find a path on a filtered graph composed only of arcs from the path we want to verify.
Atomic algorithms can be exact or heuristic. In both cases, the proposed ACG algorithm finds optimal solutions as atomic algorithms are only used to guide the search and verify constraints.
When exact, as we will show in Sec.~\ref{sec:general} for the generalization, ACG provides a better continuous relaxation.

Let's consider a set $\mathcal{A}$ of atomic algorithms that have been selected to cover all constraints of a particular GCSP problem. From an optimization point of view, the problem consists in finding a path $\tilde{p}$ of minimum cost such that this path is a feasible path for each atomic algorithm of $\mathcal{A}$. In the following, we will show that even if an atomic algorithm solves the associated problem heuristically in the end the ACG converges to the optimal solution.

Consider a network $G=(V, A)$ where $V$ is the set of nodes and $A$ is the set of arcs. Let $c: A\rightarrow \mathbb R^+$ be an arc-cost function.  $P_{\alpha}$ is the set of feasible paths associated with the atomic algorithm ${\alpha}\in \mathcal{A}$ from source $s$ to destination $t$, i.e. the set of paths respecting all constraints associated with   ${\alpha}\in \mathcal{A}$. 
Let us introduce the following variables: 
	\begin{itemize}
		\item $x_{a} \in \{0,1\} $: equals to 1 if the arc $a$ belongs to the solution path $\tilde{p}$, 0 otherwise. For each $a\in A$.
		\item $y_{p}^{\alpha} \in \{0,1\} $: equals to 1 if the path $p$ is selected as a feasible path by atomic algorithm ${\alpha}$, 0 otherwise. For each ${\alpha}\in \mathcal{A}$, and $p\in P_{\alpha}$.
	\end{itemize}

Variables $x$ link the path variables $y$ provided by atomic algorithms to ensure global convergence towards an optimal consensus solution. Let $ACG\text{-}M$ be the linear program: 
\begin{align}
		\nonumber  &  &  \min \quad\sum_{a\in A} c_a x_{a}\\
        \label{ACGRCSP:tree} & &\sum_{a\in \delta^+(u)} x_{a}  & \leq 1 & & \forall u\in V, \\ 
		\label{ACGRCSP:convexity} &\beta_{\alpha}:  & \sum_{p\in P_{\alpha}}  y_{p}^{\alpha}  &= 1 & & \forall {\alpha}\in \mathcal{A}, \\
		\label{ACGRCSP:link} &\gamma_{a,{\alpha}}:  & x_a - \sum_{p\in P_{\alpha}:a\in p}  y_{p}^{\alpha}& \geq 0 & & \forall {\alpha}\in \mathcal{A}, ~ \forall a \in A.\\
       \label{ACGRCSP:xint} & & x_a \in\{0,1\} & & & \forall a\in A\\
       \label{ACGRCSP:yint} & & y_{p}^{\alpha} \in \{0,1\}& & & \forall {\alpha}\in \mathcal{A}, \forall p\in P_{\alpha}
\end{align}
where $\delta^+(u)$ is the set of outgoing arcs of node $u$. Inequalities~\eqref{ACGRCSP:tree} ensure that the path is loop-free and elementary. As variables $x$ represent a consensus from paths provided by each atomic algorithm, we need to ensure that the induced path is an elementary one (see below Prop. \ref{propo1}).  Equalities~\eqref{ACGRCSP:convexity} guarantee that for each atomic algorithm, exactly one path that respects its associated constraints is selected. Inequalities~\eqref{ACGRCSP:link} imply that if a path for one algorithm uses link $a$, the final path also uses this link. If $x_a=0$ then Inequalities~\eqref{ACGRCSP:link} impose that variables $y$ where the link $a$ is used are set to 0. Thus, if $x$ induces a path, then each $y^\alpha$, for all $\alpha\in \mathcal{A}$ induces the same path.  
Constraints ~\eqref{ACGRCSP:link} are sufficient with $\geq$ since for every atomic algorithm, only one path is selected (induced by $y$), and path induced by $x$ cannot visit a node more than once. This implies that 
paths induced by $y$ variables and $x$ variables must be the same and thus equality is not needed.   
Let us denote by $ACG$-$M$-$RL$ the linear relaxation associated with $ACG$-$M$ where integrality constraints \eqref{ACGRCSP:xint} and \eqref{ACGRCSP:yint} are replaced by $x_a\geq 0$, $\forall a\in A$, and $y_P^{\alpha}\geq 0$. Remark trivial inequalities $\leq 1$ can be removed, since they are induced by constraints \eqref{ACGRCSP:tree} and \eqref{ACGRCSP:convexity}.

\begin{proposition}
If variables $y$ are only associated with elementary paths, then variables $x$ induce an elementary path.\label{propo1}
\end{proposition}
\begin{proof}
Inequalities~\eqref{ACGRCSP:tree} ensure that in the path induced by $x$, at most one outgoing link is selected per node. 
This guarantees that no loops, connected to the path induced by $x$. 
As the link costs are positive, no optimal solution can be composed of a path and a circuit. As the path induced by $x$ must cover paths induced by $y$ due to constraints \eqref{ACGRCSP:link} this ensures that the each path induced by each $y^{\alpha}$ and $x$ induces the same path.  Thus, the solution associated with $x$ is always an elementary path. 
\hfill $\square$
\end{proof}
$ACG$-$M$-$RL$ 
can be solved with Column Generation.
Let  $\beta$ and $\gamma$ be the dual solution associated with Inequalities \eqref{ACGRCSP:convexity} and \eqref{ACGRCSP:link} of $ACG$-$M$-$RL$. To solve it, we can generate on-the-fly columns (i.e., variables $y_p^{\alpha}$). In the remaining, we describe the associated pricing problem and the Column Generation algorithm, and we propose a dedicated branching scheme to reach optimal solutions. Combining all these ingredients, we propose an efficient Branch-and-Price algorithm, called $ACG\text{-}A$. 

Remark that the goal of the proposed model is different from the traditional Column Generation models for the multi-paths or the multi-commodity flow problems. In existing approaches, the goal is to find a set of paths respecting some constraints (e.g. disjointness, capacity), whereas in our case, the goal is to find one path respecting an heterogeneous set of constraints. 

\subsection{Pricing and Column Generation}

$ACG$-$M$ has an exponential number of variables and thus we cannot enumerate them. Therefore, we consider a Restricted Master Problem (RMP) associated with $ACG$-$M$-$RL$  where a limited number of variables (i.e., columns) are involved. At initialization, before any column has been generated, we consider a dummy variable for each atomic algorithm with a huge cost in the objective function and for which the associated path is empty. These dummy variables are needed due to the Equalities \eqref{ACGRCSP:convexity}, otherwise, the solution 
of the RMP is unfeasible. After some iterations of the algorithm, these dummy variables will be replaced in the RMP by the new generated columns/variables ($y_p^\alpha$).

Let us consider a dual solution of an RMP. 
The pricing problem, for a given atomic algorithm $\alpha$, consists in finding an elementary path $p$ with negative reduced cost, i.e. such that $\sum_{a\in p}  {\gamma}_{a,\alpha}- {\beta}_\alpha<0$. 
Remark that the atomic algorithm $\alpha$ is able to find an elementary path with minimum cost. Therefore, we can call it directly to solve the pricing problem considering dual values ${\gamma}_{a,\alpha}$ as arc costs, without knowing which kind of constraints it handles. According to this remark, we can easily deduce a Column Generation algorithm. It alternates between two phases, solving the RMP and solving the pricing problems to add new columns/variables. When no paths with negative reduced cost can be found by any of the atomic algorithms, the procedure stops. The pricing problem solved by each atomic algorithm can be itself an NP-hard problem. In this case, the linear relaxation cannot be solved in polynomial time.

\subsection{Branching Scheme}
\label{sec:optim:branch}

In this section, we describe the whole framework. Fig.~\ref{fig:framework} shows each step of the framework and how they interact together.
\begin{figure}[t]
    \centering
        \includegraphics[width=0.9\linewidth]{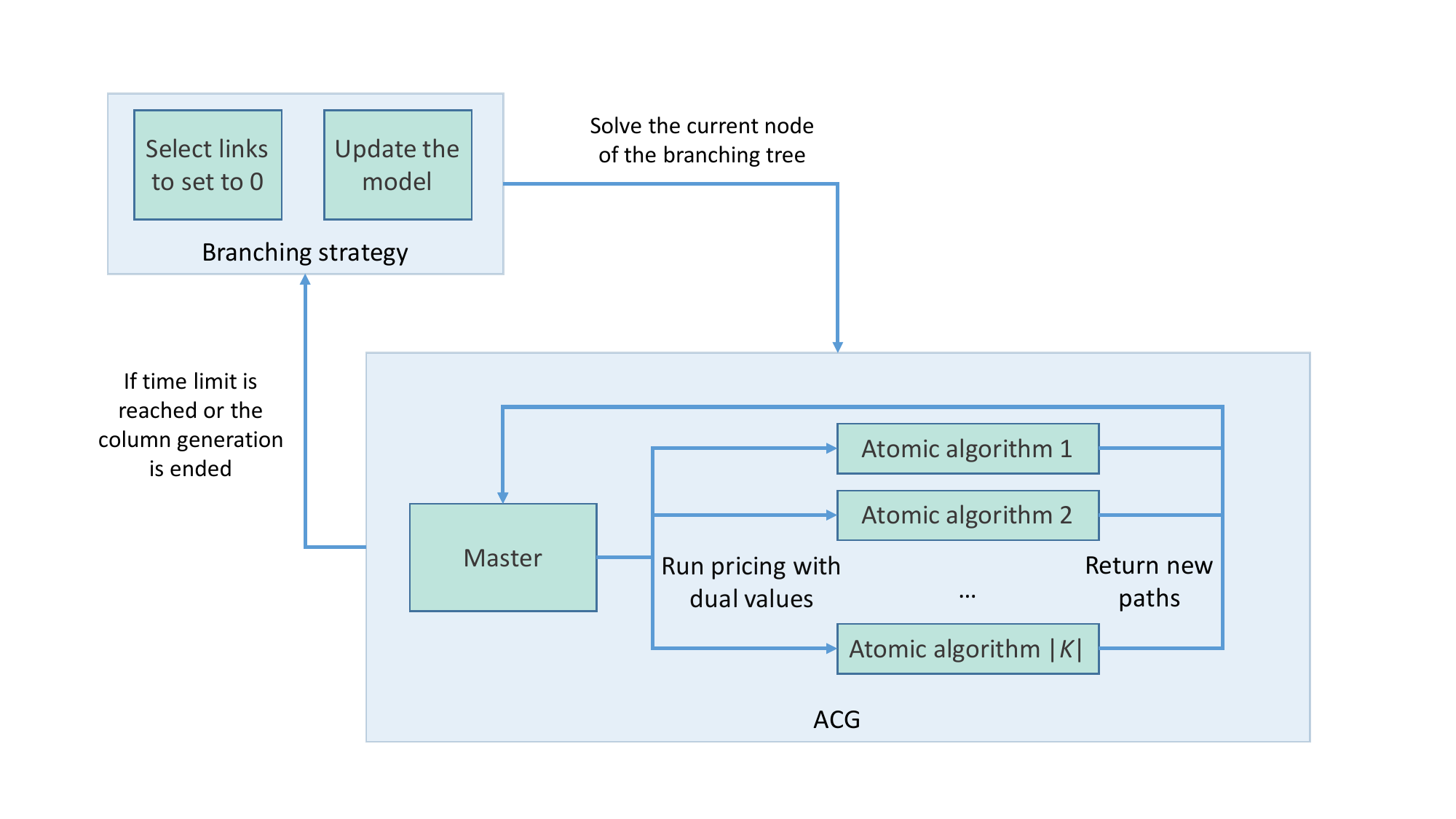}
        \caption{Framework description, where the upper block represents the branching strategy and the bottom block represents the column generation to solve the ACG model.}
        \label{fig:framework}
\end{figure}

To derive optimal integer solutions, a classical Branch-and-Bound~\cite{morrison2016branch} can be applied on top of Column Generation. However, a major issue when solving $ACG$-$M$-$RL$ comes from the time it takes to certify optimality. While we can obtain the optimal solution after a few iterations of Column Generation, we have to wait for the RMP to verify optimality. Therefore, as we need to compute a lot of these relaxations when branching, computational time can explode. 

To overcome this limitation, we designed a tailored branching scheme that is both taking advantage of solving $ACG\text{-}M$ to provide a good relaxation in a short amount of time and the possibility to call atomic algorithms to opportunistically drive the search and verify feasibility. More specifically, we call 
$$
p,l  \gets \textbf{ACG-Solve($\bar{A},T_{acg}$)}
$$
the function solving $ACG$-$M$-$RL$, 
within a time limit $T_{acg}$ on a subgraph composed of the set of arcs $\bar{A}$ ($\bar{A} \subseteq A$). As outputs, $p$ denotes a feasible path extracted from the column generation (if any), and $l$ is a dual solution-based lower bound, called, the Lagrangian bound \cite{lubbecke2010column}. Remark that $p$ can be empty if atomic algorithms do not generate feasible paths during the column generation algorithm. When a path is generated by a pricing problem, we test its feasibility for the other atomic algorithms and $p$ corresponds to the best feasible solution found.

Using this function, our branching scheme is designed as a dynamic programming algorithm, where the path is constructed and modified until every (relevant) solution is enumerated. More specifically, starting from the source $s$, it tries to construct an elementary path toward the destination $t$. Therefore, each branching node in the tree consists in choosing the next arc to take up to the destination. Based on this decision, it calls the different atomic algorithms to grasp at the earliest stage unfeasibility or optimality (see. Algo.~\ref{branch} for details). It also backtracks to improve the best solution and ensure a complete branching.

As calling \textbf{ACG-Solve} and atomic algorithms may take time, we chose to allocate them very short time limits. In practice, if an atomic algorithm does not offer this option, one can just use a timeout after which we do not expect a response anymore. When this exception occurs, we cannot get any \textit{certificates} on optimality and feasibility, which means that we cannot infer lower bounds or (un)-feasibility. Nevertheless, any feasible solution (for at least one atomic algorithm) may provide sufficient insights for the direction to take. On top of that, it is possible to extract a lower bound from \textbf{ACG-Solve} even if the time was not enough to return the optimal solution, using the Lagrangian bound. Another way to speed up \textbf{ACG-Solve} is to reuse all generated columns from previous runs. In our framework, we always consider the same model at each node of the branching tree. All generated paths are kept in the whole branching tree. The set $\bar{A}$ is managed by setting $x_a=0$ for each $a\in A\setminus \bar{A}$, this ensures that pricing problems do not change over iterations. 
In the following, we defined a heuristic to be \emph{non-trivial} if it is able to check if a path is feasible. 

\begin{proposition}
    The ACG-M with this branching scheme reaches the optimal solution even if a non-trivial heuristic is used for the pricing problems.
\end{proposition}
\begin{proof}
If $x$ is integer then it induces a path $p$ and a set of cycles not connected to $p$. As the objective function minimizes the cost of the solution, this implies that no cycle can be selected in a node of the branching tree. We need to consider at least a non-trivial heuristic for each pricing able to check if a path is feasible or not for its associated constraints $\alpha\in \mathcal{A}$. As the branching strategy consists in filtering links, we consider three cases on the subgraph $H=(V, \bar{A})$:
\begin{itemize}
    \item if $H$ is exactly a path between $s$ and $t$ then only the feasibility check is called,
    \item if $H$ does not contains a path from $s$ to $t$, i.e. $s$ and $t$ are disconnected, then the node is pruned,
    \item otherwise, continue the branching procedure.
\end{itemize}
Since the branching is complete, the result follows.
\hfill $\square$    
\end{proof}

\subsubsection{State(s) and functions}

Any state $B$ is composed of the tuple $B=(B.c,B.p,B.p^+,B.\bar{A},B.l)$, where $B.c$ represents the current cost of the solution and $B.p$ is the current (partial) path of the solution composed of a sequence of arcs. $B.p^+$ is a feasible elementary path from $s$ to $t$. It is opportunistically updated during the search as a path verified by all atomic algorithms, with respect to the current sub-path $B.p$ (hence, $B.p \subseteq B.p^+$). Finally, $B.\bar{A} $ is the set of remaining arcs that are still eligible to be chosen and $B.l$ denotes the lower bound for $B$, given by the previous choices on $B.p$. Note that, later in the paper, the function \textbf{Last} returns the last node of path $B.p$, or $s$ if the path is empty.

We will assume that  atomic algorithms provide the following functions for $\alpha \in \mathcal{A}$: 
$$
p,\text{opt},\text{unfeas}  \gets \textbf{Atomic$_\alpha$}(\bar{A},T_{atomic}) 
$$
where $p$ denotes a feasible path (empty if none has been found), and, "\text{opt}" (resp. "\text{unfeas}") equals \textit{true} if the algorithm can certify the optimality (resp. unfeasibility) of the returned path and \textit{false} if it cannot.
$\bar{A}$ denotes the set of eligible arcs that can be chosen in the returned solution, and $T_{atomic}$ is the time limit allocated to the atomic algorithm.

In practice, atomic algorithms may not support time limits and the option to work on a subset of arcs. A simple workaround to the first is to set a timeout in the function \text{\it waiting} for the answer. The second can be emulated by simply putting infinite costs on arcs that do not belong to  $B.\bar{A}$.

We denote \textbf{C($p$)} the cost evaluation function of path $p$, and $\argmin_\textbf{C}$ the function returning the element of minimum cost.

Finally, we call  
$$
\bar{A'}\gets \textbf{Filter($\bar{A},a$)} 
$$
the process that returns the updated set of eligible arcs $\bar{A'}$ that can be chosen without violating elementary path constraints, considering the initial set $\bar{A}$ from which we set an arc $a\in \bar{A}$.

\begin{algorithm}[hbt!]
\caption{Branching Scheme}\label{branch}
$B^0 \gets (c = 0, p= \emptyset, p^+= \emptyset, \bar{A}=A,l=0)$\;
$B^* \gets  (c=\infty, p = \emptyset, \bar{A}=A,l=0)$ \;
$p', l' \gets $\textbf{ACG-Solve($A$,$T_{acg}$)}\;\label{line:root}
\If{$p'\neq \emptyset$}{
$B^0 \gets (c = 0, p= \emptyset, p^+= p', \bar{A}=A,l=l')$\;
$B^* \gets  (c=\textbf{C}(p'), p = p', \bar{A}=A,l=l')$ \;
}
\text{PriorityQueue }$\gets \{B^0\}$\;
\While{\textrm{\text{PriorityQueue}}$ \neq \emptyset$}
    { 
    $B \gets$\textbf{Pop(}\text{PriorityQueue}\textbf{)}\;
    $n \gets $\textbf{Last(}$B.p$\textbf{)}\;
    \For{$a\in \delta^+(n) \cap B.\bar{A}$}{
        $\bar{A'} \gets\textbf{Filter(} B.\bar{A},a\textbf{)}$\;
        $B' \gets (B.c +c_a, B.p \cup \{a\},B.p^+ ,\bar{A'},B.l)$\;
        \If {$a \notin B.p^+$}{ \label{line:notinBplus}
            \textbf{Update($B',B^*$)}\;
        }
        \If {$B'.l<B^*.c$}{
            \textbf{Push(}\text{PriorityQueue},$B'$\textbf{)}
        }
   }
}
\Return $B^*$
\end{algorithm}

\begin{algorithm}
\caption{Update($B,B^*$)}\label{update} 
$B.l \gets B.c + \textbf{Shortest path value}(\textbf{Last}(B.p),t)$ \Comment*[r]{computed in the filtered graph $G=(V,\bar{A})$}
$c^+ \gets \infty$ \Comment*[r]{Best direction's cost}
\For{$\alpha \in \mathcal{A}$}{
    $p$,\text{opt},\text{unfeas}  $\gets \textbf{Atomic}_\alpha(B.\bar{A},T_{atomic})$\;
    \If{\textrm{unfeas}}{     \label{certif:start}
        $B.c \gets \infty$ \;
        $B.l \gets \infty$ \;
        \Return \Comment*[r]{Unfeasibility detection}
    }
    \If{\textrm{opt}}{
        $B.l \gets \max \{B.l,\textbf{C}(p)\}$\Comment*[r]{New lower bound}
        \If{\textrm{feasible(}$p$,$\mathcal{A}\setminus \{\alpha\}$\textrm{)}
        }{
            $B^*.p \gets \argmin_\textbf{C}\{ B^*.p,p\}$ \;
            $B^*.c \gets \textbf{C}(B^*.p)$
            \Comment*[r]{Update best}
            \Return \Comment*[r]{Local opt. detection}
        }
    } \label{certif:end}
\If{\textrm{feasible(}$p$,$\mathcal{A}\setminus \{\alpha\}$\textrm{)} and  \textrm{$\textbf{C}(p)<c^+$}}{
        $B^*.p \gets \argmin_\textbf{C}\{ B^*.p,p\}$ \;
            $B^*.c \gets \textbf{C}(B^*.p)$
            \Comment*[r]{Update best}
        $(B.p^+, c^+) \gets (p,\textbf{C}(p))$\Comment*[r]{Update direction}
    }
}
\If{$|B.\bar{A}|/|A|\leq \Gamma$}{
     $p',l  \gets $\textbf{ACG-Solve($\bar{A}$,$T_{acg}$)}\;
    $B.l \gets \max \{B.l,l\}$\Comment*[r]{Update lower bound}
    $B^*.p \gets \argmin_\textbf{C}\{ B^*.p,p'\}$ \;
    $B^*.c \gets \textbf{C}(B^*.p)$
            \Comment*[r]{Update best}
    \If{\textrm{$\textbf{C}(p')<c^+$}}{
        $B.p^+ \gets p'$\Comment*[r]{Update direction}
    }
}
\end{algorithm}

\subsubsection{Branching  (Algo.~\ref{branch})}

First, we initialize the best solution $B^*$ with an empty set and infinite cost, and the root solution $B^0$ with an empty path (meaning the source $s$ only) with cost $0$. Then, we apply \textbf{ACG-Solve} to extract lower and upper bounds (if applicable). We use the feasible solution  $B^0.p^+$ (if any) as a primary direction, as well as the lower bound $B^0.l$. Then, $B_0$ is pushed to a PriorityQueue to bootstrap the search. Until this queue is empty, we \textbf{Pop} out the solution $B$ with the most promising lower bound $B.l$, and choose an arc, expanding the path with respect to the elementary path constraint, calling \textbf{Filter}, to generate a new solution $B'$. If this arc was not in a previous feasible path $B.p^+$, we \textbf{Update} the solution, which can lead to a new lower bound and (possibly) a new feasible path. Finally, we check if this direction has a chance to produce a better solution and, if so, we \textbf{Push} it inside the queue. At the end, $B^*$ is the optimal solution.

\subsubsection{Update (Algo.~\ref{update})}

This procedure intends to update the different components of $B$, regarding a set of tests, as well as the best solution found so far $B^*$. At initialization, the most trivial lower bound can be computed by completing the current path $B.p$ of cost $B.c$ with the shortest path value to $t$. We also set the best feasible path cost to infinite. 

Then, we run each atomic algorithm within the time limit, considering the filtered arcs of $B$. If one of the atomic algorithms is able to prove the unfeasibility of the instance generated by $B$, then an infinite cost is set to $B$ and the process stops. On the contrary, if one of the atomic algorithms generated an optimal solution (and is able to certify it), then the lower bound $B.l$ is updated consequently. Moreover, if this solution is feasible for all other atomic algorithms, this solution is proved to be optimal for this branch and we only have to update the best solution $B^*$, before it is discarded. 

If a feasible solution is produced by at least one atomic algorithm, then we keep the one with the lowest cost; it becomes the most promising feasible path $B.p^+$ to carry on for this branch. Finally, if the proportion of remaining arcs $B.\bar{A}$ over the initial ones $A$ is sufficiently low (lower or equal to $\Gamma$, the ratio parameter), we compute \textbf{ACG-Solve} to update (if relevant) lower bound $B.l$, promising path $B.p^+$ and global solution $B^*$. This threshold $\Gamma$ is designed to run \textbf{ACG-Solve} only when the graph is sufficiently reduced, i.e., when it performs best.

On lines 10 and 14 the feasible($p$,$\mathcal{A}'$) function returns true if the path $p$ is feasible for each atomic algorithm of the set $\mathcal{A}'$. The function feasible($p$,$\mathcal{A}'$) checks the feasibility of the path $p$ by using the feasible check of each atomic algorithm $\mathcal{A}'$.

It is interesting to note that, even if atomic algorithms are heuristic, this branching scheme remains optimal. In this case, atomic algorithms would return \textit{false} for both \textit{unfeas} and \textit{opt} certificates, except when $\bar{A}$ is only composed of a single path (evaluation capability). Although the branching scheme remains complete, getting these certificates can reduce the global computational time, as we will see in Sec.~\ref{sec:results}.

\section{Generalized Decomposition}
\label{sec:general}

We now present a generic decomposition, called \text{$Generic$-$ACG$}, from which we derived ACG. It is based on a Dantzig-Wolfe decomposition over a consensus-based model where some constraints and variables are duplicated to ensure
convergence among sub-problems.
We analyze, theoretically, its properties and introduce a Branch-and-Price algorithm to solve it optimally.

The advantage of $Generic$-$ACG$ is that each pricing can be composed of all structural constraints and one or multiple additional constraints. Therefore, sub-problems can be solved with a dedicated algorithm working on the fundamental structure of the solution. A consensus solution is found among solutions of sub-problems.

\subsection{Problem Definition}

Let us define the combinatorial optimization problem $\Pi$ as follows: consider a basic finite set $E$, a cost function $c: E\rightarrow \mathbb R$ and a family $\mathcal{F}$ of subsets of $E$, find a set $F\in \mathcal{F}$ of minimum cost, i.e., $\min\{ \sum_{e\in F} c_e: F\in \mathcal{F}\}$. 
An integer program associated with problem $\Pi$ needs to define a set of binary variables $x_e$ which equal 1 if element $e$ is selected in the solution and 0 otherwise, for each $e\in E$. 

Let $S$ and $K$ be the sets of structural and additional constraints, respectively, $\Pi$ is equivalent to the $Compact$ binary program :  
\begin{align}
& \label{comp:obj} \min \sum_{e\in E} c_ex_e \\
& \label{comp:ct_base}  \sum_{e\in E} a_{e,i} x_e\geq b_i & & \forall i\in S,\\
& \label{comp:ct}  f^j(x)\geq 0 & & \forall j\in K.\\
& \label{comp:int} x_e\in \{0,1\} & & \forall e\in E
\end{align}
where  \eqref{comp:obj} is the objective function that minimizes the solution cost. Inequalities \eqref{comp:ct_base} represent the structural constraints and \eqref{comp:ct} the additional constraints. The function $f^j(x)$ has no specific form as it may be a black box, a non-linear function, etc.
Let us denote by $Compact$-$RL$ the continuous relaxation of the $Compact$ model where \eqref{comp:int} are replaced by $x_e\geq 0$ and $x_e\leq 1$, for all $e\in E$. \\
In Sec.~\ref{sec:optim}, $S$ represents the set of the shortest path constraints (flow conservation and sub-tour elimination), whereas, $K$ represents the union of all constraints associated with all atomic algorithms.

\subsection{Generic ACG Model}
We present a reformulation as another binary program with an exponential number of variables. The aim of this model is to strengthen the continuous relaxation and provide the ability to use atomic algorithms that are efficiently designed for a subset of additional constraints. Note that, if we apply a classical Dantzig-Wolfe decomposition on problem $\Pi$, by keeping Constraints \eqref{comp:ct_base} in the master problem, then all additional constraints \eqref{comp:ct} belong, together, to the same pricing problem (since they share common variables $x$). To consider one pricing problem for each additional constraint we need to duplicate constraints and variables as shown below. 
 
First, let us define a consensus-based model obtained from 
the compact model associated with the problem $\Pi$ by removing \eqref{comp:ct} and adding new variables $x^j_e\in \{0,1\}$ for $j\in K$ and $e\in E$ and the following constraints: 
\begin{align}  
& \label{comp_ext:ct_linking}  x_e - x^j_e = 0 & & \forall j\in K, e\in E,\\
& \label{comp_ext:ct_ressources_2_base}  \sum_{e\in E} a_{e,i} x^j_e\geq b_i & &  \forall j\in K, i\in S,\\
& \label{comp_ext:ct_ressources_2}  f^j(x^j)\geq 0 & & \forall j\in K.
\end{align} 
Equalities \eqref{comp_ext:ct_linking} guarantee that the global solution satisfies all additional constraints. Constraints \eqref{comp_ext:ct_ressources_2_base} are redundant as they are equivalent to \eqref{comp:ct_base}, but they will be useful after the decomposition to enhance the continuous relaxation. Constraints \eqref{comp_ext:ct_ressources_2} represent the additional constraints using the new variables. 

For $j\in K$, let $\mathcal{F}_j$ be the set of solutions respecting \eqref{comp_ext:ct_ressources_2_base} and \eqref{comp_ext:ct_ressources_2} associated to $j$. Note that $\mathcal{F}\subseteq \mathcal{F}_j$. For GCSP (see definition in Sec.~\ref{sec:system}) the set $\mathcal{F}_j$ represents the set of all solutions of PCE $j$.
Let us also define, for each $j\in K$, a new variable $y_F^j$, that equals $1$ if $F\in \mathcal{F}_j$ is selected in the solution, $0$ otherwise.  

Now, let us introduce the \textit{$Generic$-$ACG$ Model ($Generic$-$ACG$-$M$)}, applying Dantzig-Wolfe over the consensus-based model by keeping \eqref{comp:ct_base} and \eqref{comp_ext:ct_linking} in the master problem:     
\begin{align}
& & &\label{ACG:Obj} \min \sum_{e\in F} c_ex_e \\
& & & \label{ACG:S}   \sum_{e\in E} a_{e,i} x_e \geq b_i & & \forall i\in S,\\ 
&\beta_j: & &  \label{ACG:Conv} \sum_{F\in  \mathcal{F}_j} y_F^j=1 & & \forall j\in K,\\
&\gamma_{e,j}: & &  \label{ACG:Rj} x_e - \sum_{F\in  \mathcal{F}_j: e\in F} y^j_{F} = 0& &\forall j\in K, \forall e\in E.\\
& & & \label{ACG:x} x_e\in \{0,1\}& & \forall e\in E,\\
& & & \label{ACG:y} y^j_F\in \{0,1\} & & \forall j\in K, \forall F\in  \mathcal{F}_j.
\end{align}  
 \eqref{ACG:Obj}-\eqref{ACG:S} are equivalent to \eqref{comp:obj}-\eqref{comp:ct_base}. Equalities \eqref{ACG:Conv} ensure that for each additional constraint, exactly one solution is selected. Equalities \eqref{ACG:Rj} guarantee all additional constraints managed by a pricing problem fit with the original variables $x$. 
Let us denote by $Generic$-$ACG$-$M$-$RL$ the linear relaxation where integrality constraints \eqref{ACG:x} and \eqref{ACG:y} are replaced by $x_e\geq 0$, for all $e\in E$, and $y^j_F\geq 0$, for all  $j\in K$ and $F\in  \mathcal{F}_j$.
Remark that due to constraints \eqref{ACG:Conv} it is not necessary to add $y^j_F\leq 1$, for all $j\in K$. We can deduce, due to the constraints \eqref{ACG:Rj} that if $y$ variables are bound by 1 then $x$ variables are also bounded by 1.

\paragraph{Models analysis}
\begin{proposition} \label{ACG_Compact_RL}
\label{prop_relax}
$Generic$-$ACG$-$M$-$RL$ is stronger than or equal to $Compact$-$RL$.   
\end{proposition}
\begin{proof}
Let $(x^*, y^*)\in [0,1]$ be a solution of the linear system of \eqref{ACG:S}-\eqref{ACG:Rj} and trivial inequalities. It is easy to see that $ x^*$ satisfies constraints \eqref{comp:ct_base} associated with $\mathcal S$. Moreover, by \eqref{ACG:Conv}, $x^*$ is a convex combination of all solutions in $\mathcal F_j$. Therefore, $x^*$ satisfies all Constraints \eqref{comp:ct}. Hence, $ x^*$ is a feasible solution for the $Compact$-$RL$, and thus the $Generic$-$ACG$-$M$-$RL$ is at least equal to the $Compact$-$RL$. Now let us show an example where $Generic$-$ACG$-$M$-$RL$ is stronger than $Compact$-$RL$. To show that the 
$Generic$-$ACG$-$M$-$RL$ may be strictly stronger than $Compact$-$RL$, we consider an instance of the RCSP with two upper-bound constraints. Consider a graph $G=(\{s,u,v,t\}, \{su, sv,ut,vt\})$ and two resources with upper-bounds $U_1=12$ and $U_2=9$. Edges $su$ and $ut$ have, both, a consumption of $3$ for each resource, and cost of $2$. Edges $sv$ and $vt$ have, both, a consumption of $6$, for each resource, and a cost of $1$.
Each pricing problem in the $Generic$-$ACG$-$M$-$RL$ represents a shortest-path problem with one capacity constraint. For resource 1, two paths can be generated: $\{su,ut\}$ and $\{sv,vt\}$. For resource 2, one path can be generated: $\{su,ut\}$. By \eqref{ACG:Rj}, it follows that $y^1_{\{su,ut\}} = y^2_{\{su,ut\}} = 1$ and $x_{su}=x_{ut}=1$ is the optimal solution with a value of $4$. For the $Compact$-$RL$, the solution $x_{su}=x_{ut}=x_{sv} = x_{vt} = 0.5$ satisfies the conservation flow constraints at each node and the resource capacities. Moreover, its objective value is $3$ which is lower than the optimal value of $Generic$-$ACG$-$M$-$RL$. Therefore, in this example, $Generic$-$ACG$-$M$ has a strictly stronger linear relaxation than $Compact$ model. This ends the proof. \hfill $\square$
\end{proof}

Let us consider the classical Dantzig-Wolfe decomposition on the $Compact$ model given by the $DW$-$M$ model:
\begin{align}
& \label{DW:obj} \min \sum_{e\in E} c_e\sum_{F\in \mathcal{F}:e\in F}y_F \\
& \label{DW:ct_conv} \sum_{F\in \mathcal{F}}y_F=1&&\\ 
& \label{DW:ct_base}  \sum_{e\in E} a_{e,i} \sum_{F\in \mathcal{F}:e\in F}y_F\geq b_i & & \forall i\in S,\\
& \label{DW:int} y_F\in \{0,1\} & & \forall F\in \mathcal{F}
\end{align}
we denote by $DW$-$M$-$RL$ its linear relaxation, where integrality constraints \eqref{DW:int} are replaced by $y_F\geq 0$ for all $F\in \mathcal{F}$. Clearly $y_F\leq 1$ for all $F\in \mathcal{F}$ are induced by constraints \eqref{DW:ct_conv}.
\begin{proposition} 
$Generic$-$ACG$-$M$-$RL$ may be strictly stronger than $DW$-$M$-$RL$.
\end{proposition}
\begin{proof}A RCSP instance with two upper-bound constraints is enough to show the proposition. Let consider the same example as in the proof of Prop. \ref{ACG_Compact_RL}. As shown previously, the optimal value of $Generic$-$ACG$-$M$-$RL$ is $4$. 
For the $DW$-$M$-$RL$ the pricing problem is a multi-knapsack problem with two knapsack constraints. In the optimal solution of the linear relaxation, two generated columns have values $y_{\{su,vt\}} = y_{\{sv,ut\}} = 0.5$. This corresponds to solution $x_{su}=x_{ut}=x_{sv} = x_{vt} = 0.5$ in the compact model. Therefore, in this example, $Generic$-$ACG$-$M$ has a strictly stronger linear relaxation than $DW$-$M$-$RL$. This ends the proof.   \hfill $\square$
\end{proof}

\paragraph{Problem Properties}  
The equality constraint in \eqref{ACG:Rj} results in both negative and positive dual costs. This introduces complexity in optimization problems, particularly when dealing with mixed-sign costs. For example, in the GCSP, the pricing problems are variants of the shortest-path problem that requires positive costs. In the following, we show how to avoid mixed-sign dual costs.
\\
\begin{definition}\label{cond1}
Optimization problem $\Pi$ is said to be \textit{union-free}, if and only if, for each pair of constraints $j_1, j_2 \in K$, there does not exist a solution $F_{j_1}\in \mathcal{F}_{j_1}$ and $F_{j_2}\in \mathcal{F}_{j_2}$ such that $F_{j_1} \neq F_{j_2}$ and $(F_{j_1}\cup F_{j_2})\in \mathcal{F}$.
\end{definition}
For example, GCSP is union-free since the union of two different paths from the same source and destination cannot be a path. On the other hand, the Knapsack problem is not union-free. Indeed, it is easy to see that the union of two knapsack solutions may satisfy the knapsack constraint.

\begin{proposition}\label{prop:1}
If problem $\Pi$ is union-free then, in $Generic$-$ACG$-$M$, Equalities \eqref{ACG:Rj} can be converted to $"\geq"$.
\end{proposition}
\begin{proof}
Suppose that in $Generic$-$ACG$-$M$, Equalities \eqref{ACG:Rj} cannot be converted to $"\geq"$. It follows that, after conversion, there exists a solution $(x,y)$ that satisfies $\eqref{ACG:S}-\eqref{ACG:Rj}$ but for at least one $j\in K$, $x$ violates \eqref{comp:ct}. Clearly, by \eqref{ACG:Rj}, solutions of two different atomic algorithms are included in the solution associated with $x$. This implies that $\Pi$ is not union-free and the result follows.\hfill $\square$
\end{proof}
 
\begin{proposition}\label{propnogood}
If problem $\Pi$ is not union-free, equalities \eqref{ACG:Rj} can be converted to $"\geq"$ after adding the following constraints
\begin{align}
\sum_{e\in  {F'}}x_e \leq | {F'}|-1 + \sum_{e\in E\setminus  {F'}}x_e,\qquad  {F'} \in  {\mathcal{F'}} \label{no_good_cut}
\end{align}
where $ {\mathcal{F'}}$ is the set of unfeasible solutions, i.e., $ {\mathcal{F'}}=\{E':E'\subseteq E, E'\notin\mathcal{F}\}$. 
\end{proposition}
\begin{proof}
Inequalities \eqref{no_good_cut} are the well-known no-good cuts~\cite{bockmayr2003detecting}, thus they ensure that the solution provided by $x$ is valid. Indeed, for any unfeasible integer vector of $x$ the generated no-good-cut only discards this integer vector $x$. It remains valid for any other integer solution in the original model. Thus, only a feasible integer vector can be found at the end of the optimization.  \hfill $\square$
\end{proof}

Note that sub-tour elimination constraints \cite{subtour} are strengthened no-good cuts \eqref{no_good_cut} for removing cycles in the solution. Indeed for cycles $\sum_{e\in E\setminus {F'}}x_e =0$ and cuts can be written as follows
\begin{align}
\sum_{e\in  {F'}}x_e \leq | {F'}|-1,\qquad  {F'} \in  {\mathcal{F'}} \label{no_good_cutSub}
\end{align}
where $\mathcal{F'}$ represents all sub-tour in the graph. Moreover, no-good cuts can be further improved. For a given problem, analyzing structural constraints allows to strengthen them. For instance, in $ACG\text{-}M$, we can remark that the flow conservation and the sub-tour elimination constraints are replaced by Constraints \eqref{ACGRCSP:tree}.

\subsection{Branch-and-Price Algorithm} 

As $Generic$-$ACG$-$M$ involves an exponential number of variables, 
a Branch-and-Price algorithm ($Generic$-$ACG$-$A$) can solve it, using Column Generation and Branch-and-Bound.
In this section, we refer for \textit{Restricted Master Problem (RMP)} the $Generic$-$ACG$-$RL$, restricted to a subset of variables $y^j$, for each $j\in K$.\\
\paragraph{Column Generation}
It, iteratively, consists in solving the RMP and the pricing problems.

    \textbf{Restricted Master Problem:} It combines solutions of all sub-problems (atomic algorithms) in order to find an optimal consensus solution that respects all additional constraints. At each resolution, we obtain the dual solutions that will drive the sub-problems to converge. 
    
    \textbf{Pricing problem:} For each $j\in K$, this problem looks for a set $F\in \mathcal F_j$, such that the associated variable $y^F_j$ is able to improve the linear relaxation of RMP. Let $\beta$ and $\gamma$ be the dual solution associated to Constraints \eqref{ACG:Conv} and \eqref{ACG:Rj}. 
    For each $j\in K$, the pricing problem is equivalent to solving  $\min\{ \sum_{e\in F } {\gamma}_{e,j}- {\beta}_j :F \in \mathcal{F}_j\}$ and if the value is negative then the associated variable is added to the current RMP.  
    
 \paragraph{Branching Scheme} 
If the solution of the RMP is not integer, then we need to branch to obtain an integer solution. In the following proposition, we show that branching only on variables $x$ is enough. Branching on $x$ allows us to, drastically reduce the size of the branching tree, compared to branching on $y$, as they are in polynomial number.

\begin{proposition}\label{intx1}
If $\Pi$ is union-free, and vector $x$ is integer, then 
so is $y$.
\end{proposition}
\begin{proof} 
Let us consider vectors $(x^*, y^*)$ satisfying Constraints \eqref{ACG:S}-\eqref{ACG:x} and let $F^x\subseteq E$ (resp. $F^y\subseteq E$) be the subset of $E$ associated to $x$ (resp. to one variable $y$), with nonzero values. First, we show that $F^y\subseteq F^x$. Let us suppose the contrary, that is $F^y\not\subseteq F^x$ where $e\in F^y\setminus F^x$. As $e\notin F^x$, $x^*_e=0$. By Constraints \eqref{ACG:Rj}, it is easy to see that variable $y^j_{F^y} = 0$. Now, suppose that $x^*$ is integer and $y^*$ is fractional. By \eqref{ACG:Conv}, there exists $j\in K$, $F^y_1,F^y_2\in \mathcal F_j$, with $F^y_1\neq F^y_2$, such that $y^*_{F^y_1}>0$ and $y^*_{F^y_2}>0$. Moreover, equalities \eqref{ACG:Rj} implies that $F^x$ is the union of $F^y_1$ and $F^y_2$. This implies that problem $\Pi$ cannot be union-free, contradiction. 
\hfill $\square$ 
\end{proof}

Remark that due to the no-good cut \eqref{no_good_cut}, if the problem $\Pi$ is not union-free, then $x$ is a valid solution for each atomic algorithm. Thus, there exists an associated integer vector $y$. In the case where the equalities \eqref{ACG:Rj} are not replaced, since they come from the equalities \eqref{comp_ext:ct_linking} after a Dantzig-Wolfe decomposition this implies that if the vector $x$ is integer then $y$ so is.

Thanks to Prop. \ref{intx1} and previous remarks it is only necessary to branch on the $x$ variables that are polynomial in number. 
Thus, classical branching strategies can be used. 
For the implementation, we can manage the branching only by fixing some $x$ variables to 0 and have it complete. If a branching is done by setting $x$ variables to 0 then the pricing is managed by filtering elements from the original problem to work on a smallest instance. A basic branching strategy is to consider any possible binary combination of the vector of the variable $x$ and let the variable be free when the value is equal to 1. In this case, the branching is complete since all solutions are explored. 

\begin{proposition} The $Generic$-$ACG$-$A$ reaches the optimal solution even if a non-trivial heuristic is used for the pricing problems.
\end{proposition}
\begin{proof}
    
If $x$ variables are integers at some node of the branching tree and as $x$ satisfies the structural constraints \eqref{ACG:S}, each pricing problem can be used to check if the solution, given by $x$, respects the associated additional constraints or not. This approach works since the pricing problem is solved in a filtered instance, where elements are removed according to the variables $x$ equal to 0. Since the branching is complete, the result follows.   \hfill $\square$ 
\end{proof}

\section{Numerical Results}
\label{sec:results}

We now evaluate the Augmented PCE over a particular RSCP use case. It has been selected for the variety of its constraints and the fact that we can derive, for benchmarking, an effective and optimal dedicated algorithm.

\subsection{Use Case: Resource Constrained Shortest Path (RCSP)}
\label{sec:rcsp}

We consider a path computation problem with $3$ upper bounds  (jitter, loss and TE cost), $3$ ranges  (average delay, peak delay, hop count) and $1$ node inclusion. This problem can be converted to an RCSP problem. Indeed, packet loss is made additive using the log function. Node inclusion is satisfied using a specific resource called node inclusion resource, being $1$ on outgoing arcs of the inclusion node and $0$, elsewhere. Lower and upper bound of the node inclusion resource ensure that the metric over the path is exactly $1$ and thus the node must be included. For example, in a graph $G=(V,A)$, if a node $u\in V$ must be included in a path from $s\in V$ to $t\in V$, then we can consider a specific range constraint to handle the node inclusion constraint. Clearly, if the lower and upper bounds are set to 1, the metric for this resource is set to 0 for all links except the outgoing links of $u$ then the shortest path must cross exactly one outgoing link of $u$. 

More formally, let $r_a^j$ be the consumption of resource $j\in K$ associated with the link $a\in A$. The RCSP problem consists in finding the shortest path $p$ such that 
upper bound, i.e, $\sum_{a \in p}r_a^j \leq u_j$ and lower bound, i.e., $\sum_{a \in p}r_a^j \geq l_j$, constraints are satisfied. 
Note that to consider  lower bound constraints,
we must add subtour elimination constraints, so that no circuit appears in the solution, as in the formulation presented in~\cite{beasley1989algorithm}.

\subsection{Benchmark Algorithms: MultiPulse and Compact ILP Model}

To solve the RCSP model, algorithms based on Lagrangian relaxation~\cite{beasley1989algorithm, pugliese2013survey} combined with dynamic programming have been proposed. However, the associated sub-problem consists in a shortest path with negative costs (due to lower bound constraints) which must be solved optimally so that the sub-gradient algorithm provides a useful lower bound to dynamic programming, making such a solution impractical.

\textbf{MultiPulse.} As far as we know, Pulse~\cite{LOZANO2013378} outperforms existing algorithms for multiple upper-bound constraints and it is optimal. It adopts a memory-efficient dynamic programming approach using DFS (Depth First Search). For the sake of benchmarking, we generalized it to take into account lower-bound constraints. If the pruning procedures remain unchanged for the upper bounds, we define a specific DFS strategy to quickly go toward a feasible solution and satisfy lower bound constraints. From a current state $s$ with $n$ being the last node of the partial path $p$ under construction, we choose, in preference, arc $a^* \in \delta^+(n)$ to extend the path, so that: 
$$
a^*= \argmin_{a' \in \delta^+(n)} \max_{j \in K}  (l_j - \sum\limits_{a\in p}r^{j}_{a} - r^j_{a'}) 
$$
if $\max_{j \in K}  (l_j - \sum\limits_{a \in p}r^{j}_a) > 0$, otherwise, we take the shortest path direction.

In other words, the DFS strategy consists in choosing the next arc so that the worst gap of current resource consumption to lower bounds is reduced. This occurs only if at least one lower bound is not reached. Otherwise, it goes back to a classical shortest-path strategy. The efficiency of such an approach lies in the fact that lower bounds can be difficult to satisfy for some resources. On the contrary, the pruning techniques described in \cite{LOZANO2013378} are known to be very efficient to satisfy upper-bound constraints at the earliest stages.

\textbf{Compact ILP.} We also compare against a Compact model. 
Recall that, as opposed to ACG, it has full knowledge of the problem and cannot reuse any available black-box algorithm. In this case, we solve the model using IBM ILOG CPLEX 12.6 solver~\cite{cplx}.

In our implementation of ACG, as the underlying atomic algorithm (one instance for each constraint), we considered MultiPulse configured either for one range or one upper bound constraint. We also considered short time limits inside branching of $T_{acg}=500$ms and $T_{atomic}=60$ms. Table \ref{timelimits} provides a comparison with different values of the time limits. Each line provides the average execution time over the 70 instances considered in the rest of the paper. We have considered either $T_{acg}$ or $T_{atomic}$ with a value multiplied by two or divided by two. 
We remark that the setting  $T_{acg}=500$ms and $T_{atomic}=60$ms obtains the best average computational time.
\begin{table}
\begin{center}
\begin{tabular}{|c | c |  c | c|} 
 \hline
  Settings ($T_{acg}, T_{atomic}$) & Average of computational times  \\[0.5ex] 
 \hline\hline
  $500$ms, $60$ms  & \textbf{53.46ms} \\
 \hline
  $500$ms, $30$ms  & 53.51ms \\
 \hline
  $500$ms, $120$ms & 53.81ms \\
 \hline
 $250$ms, $60$ms & 54.16ms \\
 \hline
  $1000$ms, $60$ms & 54.03ms\\
 \hline
\end{tabular} 
 \caption{Impact of time limits on computational times. }
 \label{timelimits}
\end{center}
\end{table}

All implementations are in C++ on a machine with an Intel(R) Xeon(R)
CPU E5-4627 v2 at 3.30GHz and 504GB RAM, running Linux 64 bits. A maximum of $32$ threads has been used for CPLEX and ACG.  
More specifically, ACG's branching scheme was allowed to use different threads to explore simultaneously different branches, as well as for the execution of atomic algorithms. For fair comparison against MultiPulse, we also provide results with $1$ thread for ACG. Finally we use $\Gamma=0.2$ and in all cases, the global time limit is $120$s.

\subsection{Instances}

\subsubsection{Topologies}
We tested two types of instances, the first set was made on 31x31 grid topology (961 nodes, 3720 arcs) in a similar manner as RMF \cite{Go98}, as it generates realistic instances. Then we complete the experimental setup with SNDlib topologies: france, geant, germany50, giul39, janos-us, nobel-eu, sun, ta1, and zib54 
\cite{SNDlib} and zoo topologies: Abvt, Canerie, Evolink, Goodnet, HurricaneElectric, Ibm, Quest, Rediris, Xeex \cite{zootopology}. The way to generate resources and constraints is the same for each topology and is described below\footnote{All instances are available on the following public repository:\\
\url{https://github.com/MagYou/Atomic-Column-Generation}}. 

\subsubsection{Feasible instances
}
First we generated cost and resource vectors as uniform distribution in $[10,100]$. Then, from a random node, we perform a random walk that must be elementary. The random walk stops when a \textit{path size} limit parameter is reached, or no additional node can be appended without violating the elementary path property. For this path, we compute its associated resource value and we allow some variations to produce constraints. To produce non-trivial instances, we choose a variation of 20\% (an addition for upper bound and a subtraction for lower bound, if applicable). This ensures the feasibility of each instance while the optimal solution is not necessarily the original computed path. We considered 10 instances for each path size and considered for each one 3 upper bound and 3 range constraints. The node to include is randomly selected among nodes in the path computed from the random walk, not considering the source and the destination.

\subsubsection{Unfeasible instances}
We have also tested the capability to detect unfeasibility. To generate hard instances, given the previous metrics and bounds, we first computed the optimal RCSP on a given resource vector $r_1$, with a cost function corresponding to another resource vector $r_2$. Then for the latter resource, we enforce an upper bound strictly less than this optimal solution value. In this way, the unfeasibility comes from the incompatibility among the different sub-problems while a feasible solution may exist for each of them separately. For example, given a resource $r_1$ and its associated lower and upper bound $l_1,u_1$, we find the optimal path $p$ that minimizes resource $r_2$, while $l_1 \leq V_{r_1}(p) \leq u_1$ is satisfied (with $V_r(p)$ is the resource value of path $p$ on resource $r$). As $V_{r_2}(p)$ is minimal, putting an upper bound $u_2<V(r_2)$ will automatically leads to unfeasibility, that can only be detected by confronting the two resources. 

\subsection{Results}
In this section, we present the performances of each algorithm on four sets of instances. Feasible and unfeasible options applied on Grid and realistic instances.
\subsubsection{Feasible instances}
First, let us focus on the feasible instances.
\paragraph{Grid instances}
First, we extensively benchmark our different versions of $ACG$ on Grid instances. This will help us to keep the best $ACG$ version in terms of solving capabilities. \\

We have considered the following versions of the ACG method: 
\begin{itemize}
    \item ACG: the Branch\&Price algorithm where 32 threads are allowed,
    \item ACG-1: the Branch\&Price algorithm with only one thread,
    \item ACG-H: the Branch\&Price algorithm where atomic algorithms cannot certify the optimality and unfeasibility (case with heuristic atomic algorithms),
    \item ACG-R: the solution without branching (root node).
\end{itemize}

For each figure presented in this section, each dot represents the average value over 10 random instances.
\begin{figure}[h!]
    \centering
        \includegraphics[width=0.5\linewidth]{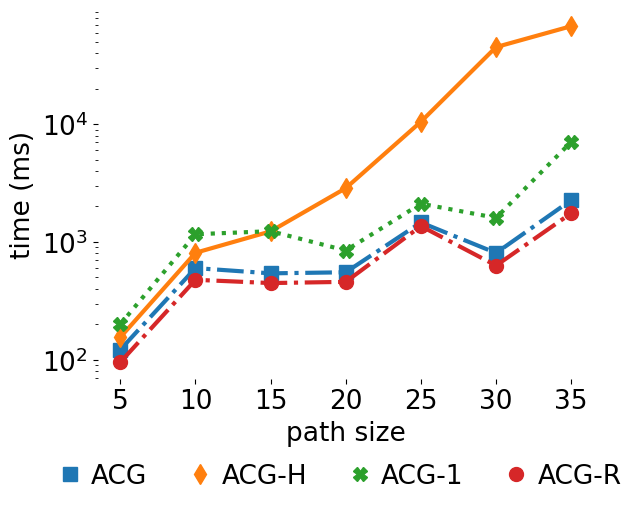}
        \caption{Computation time (ms) to compare the different options of the ACG method on feasible instances from Grid topologies.}
        \label{fig:time-02}
    \centering     \includegraphics[width=0.5\linewidth]{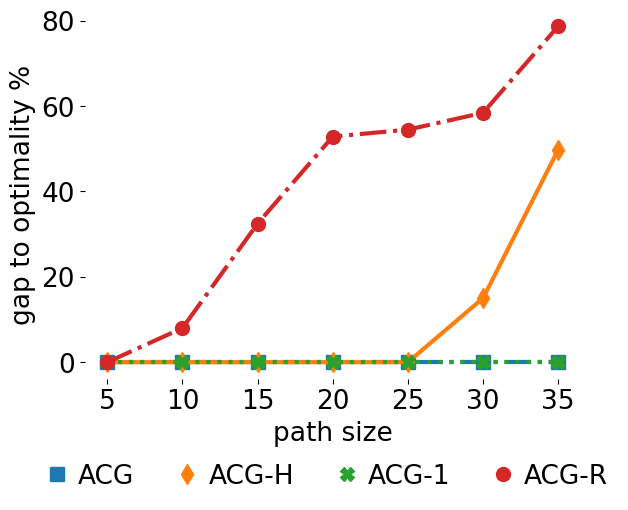}
        \caption{Optimality gap (\%) to compare the different options of the ACG method on feasible instances from Grid topologies.}
        \label{fig:gap-02}
\end{figure}

 \begin{figure}[h!]
        \centering
    \includegraphics[width=0.55\linewidth]{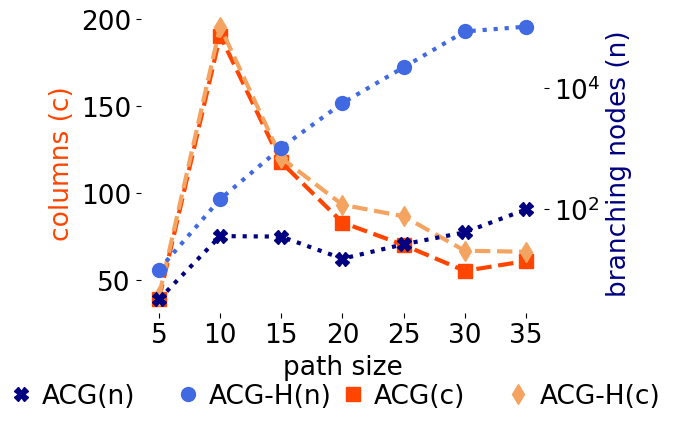}
        \caption{Columns and branches generated by the different options of the ACG method on feasible instances from Grid topologies.}
        \label{fig:items-02}
\end{figure}

\begin{figure}[h!]
    \centering
        \includegraphics[width=0.5\linewidth]{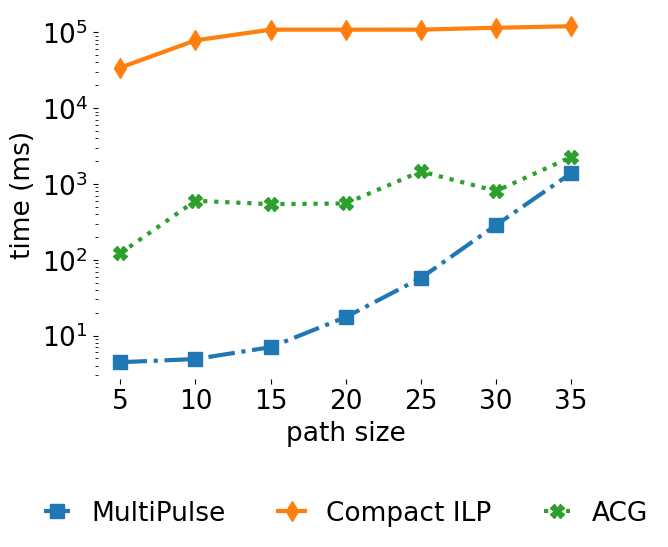}
        \caption{Computation time (ms) for comparing the three methods, state-of-the-art compact ILP, dedicated algorithm MultiPulse method, and our generic ACG method on feasible instances from Grid topologies.} 
        \label{fig:time-01}
\end{figure}
\begin{figure}[h!]
        \centering
        \includegraphics[width=0.5\linewidth]{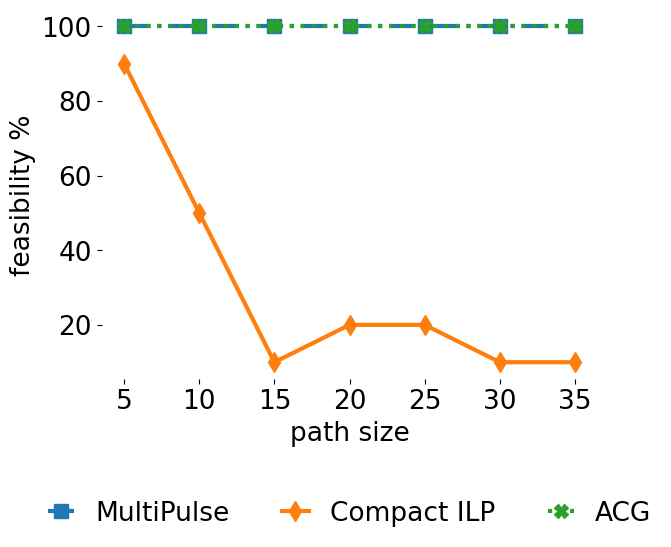}
         \caption{Feasibility (\%) for comparing the three methods, state-of-the-art compact ILP, dedicated algorithm MultiPulse method, and our generic ACG method on feasible instances from Grid topologies.}
         \label{fig:feas-01}
\end{figure}
\begin{figure}[h!]
    \centering     \includegraphics[width=0.5\linewidth]{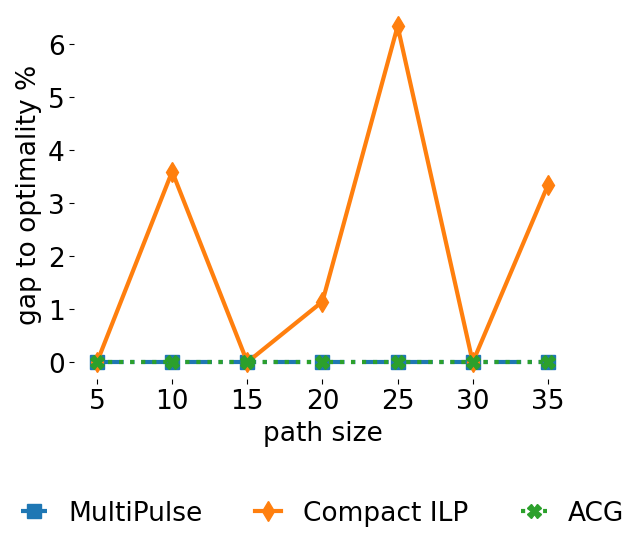}
        \caption{Optimality gap (\%) for comparing the three methods, state-of-the-art compact ILP, dedicated algorithm MultiPulse method, and our generic ACG method on feasible instances from Grid topologies.}
        \label{fig:gap-01}
\end{figure}

In Fig.~\ref{fig:time-02}, we remark that the number of threads has a non-negligible impact on the branching part, and it allows for reducing the computational time, close to ACG-R (without branching). Furthermore, when pricing problems are heuristics, i.e. ACG-H, the computational time increases since the branching strategy needs more time to prove the optimality. This shows that the quality of the solution in the root node has an impact on the branching strategy. For all instances, each ACG method finds a feasible solution. Fig.~\ref{fig:gap-02} shows the quality of the solution found. ACG and ACG-1 always found the optimal solution whereas ACG-H reaches the time limit for some instances where the path size is greater than $25$. We also remark that ACG-R (without branching) always finds a feasible solution and the quality of the solution degrades in correlation with the instance hardness i.e, when the path size increases. In Fig.~\ref{fig:items-02}, we remark that for a path size equals to $5$, the number of generated columns is small and the solution is found at the root node for the ACG method where pricing is solved to optimality. When pricing problems are solved heuristically, i.e. ACG-H, we need to branch even for a path size equals to $5$. The branching allows to guarantee optimality. We also remark that when the path size is more than $10$, the number of generated columns decreases. This is due to the time limit where ACG needs more time to find columns. An interesting observation is that, when pricing problems are solved to optimality, the number of nodes in the branching tree does not grow too much according to the difficulty of the instance. For ACG-H, the opposite can be observed since the number of nodes in the branching tree increases exponentially with the difficulty of the instance. These results confirm the strength of our model to find a good solution at the root node and drive the branching strategy.

We now focus, on the performance of ACG method in comparison with the state-of-the-art compact ILP and a dedicated algorithm to solve the problem, i.e. MultiPulse. Fig.~\ref{fig:time-01} shows the computational time of these three methods. 
First, MultiPulse is the best algorithm in terms of speed and solution's quality. 
Second, the compact ILP quickly reaches the time limit, except on small instances where it remains competitive, see Fig. \ref{fig:sndlib}. Third, we remark that instances with a bigger path size are harder to solve, and the computational time of the ACG method converges to the computational time of MultiPulse when instances become harder. In Fig.~\ref{fig:feas-01}, we remark that, for a lot of instances, the compact ILP is not able to find a feasible solution. For the instances where all methods found a feasible solution, we compare the gap to optimality in Fig.~\ref{fig:gap-01}. We remark that, when compact ILP is able to find a feasible solution, the gap to optimality can be higher than 6\%, whereas MultiPulse and ACG always find the optimal solution.

We also compared the optimal solution of the relaxations, namely ACG-R and the linear relaxation of the compact ILP. In our instances, getting the optimal solution of ACG-R was often 10 times slower than for the compact ILP, but it was most of the time equal to ACG while compact ILP's relaxation was up to 30\% of this value, confirming experimentally Proposition~\ref{prop_relax}.

\paragraph{Realistic instances}~

We now compare our best ACG setup, called ACG, 
on a batch of realistic topologies, following the same process as before for generating the resources, lower and upper bounds and node inclusion. Fig. \ref{fig:performance} shows the performance profile, that depicts the percentage of solved instances given a certain amount of time.
Although the compact ILP is clearly dominated on our Grid instances, it is globally quite efficient to solve the realistic instances, see Fig. \ref{fig:performance}, as the topologies are small. However, the chart also showcases a non negligible set among of them that has been solved in shorter amount of time by the two other algorithms. On average, the compact ILP performs well on these small instances. We can remark that Multipulse is always better than ACG but the performances are close for instances with large path size. 

\begin{figure}[h!]
    \centering
        \includegraphics[width=0.5\linewidth]{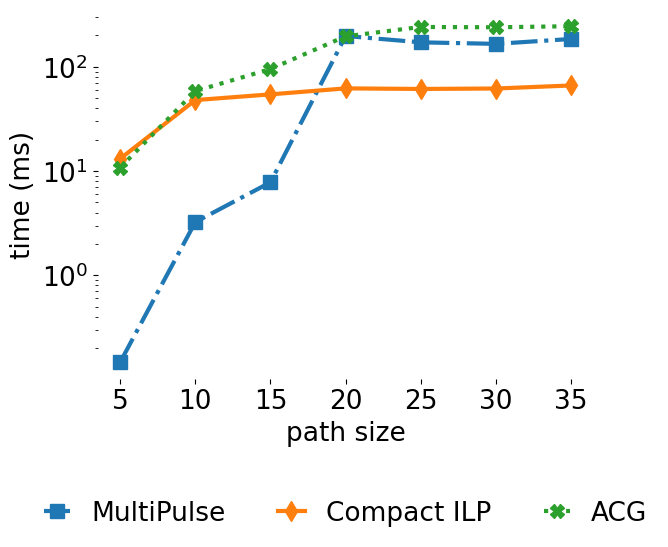}
        \caption{Computation time (ms) for comparing the three methods, state-of-the-art compact ILP, dedicated algorithm MultiPulse method, and our generic ACG method on feasible instances from realistic topologies.}
        \label{fig:sndlib}
\end{figure}

\begin{figure}
    \centering
    \subfloat[\centering Feasible instances from Grid topologies.]{{\includegraphics[width=0.45\linewidth]{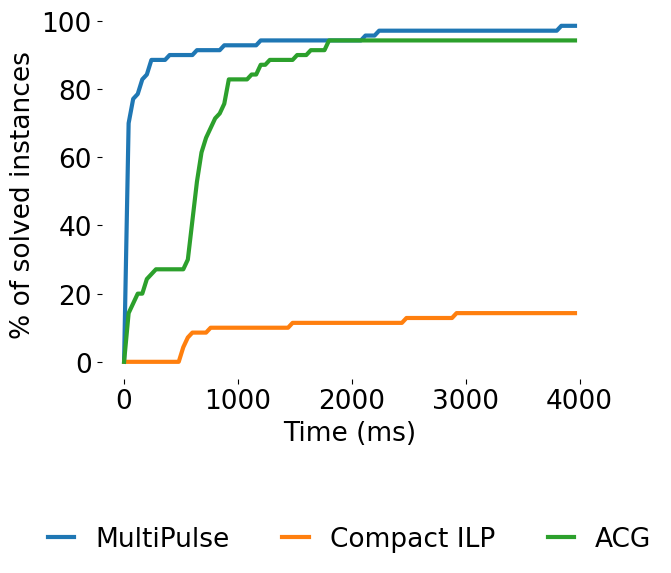} }}%
    \qquad
    \subfloat[\centering Feasible instances from realistic topologies.]{{\includegraphics[width=0.45\linewidth]{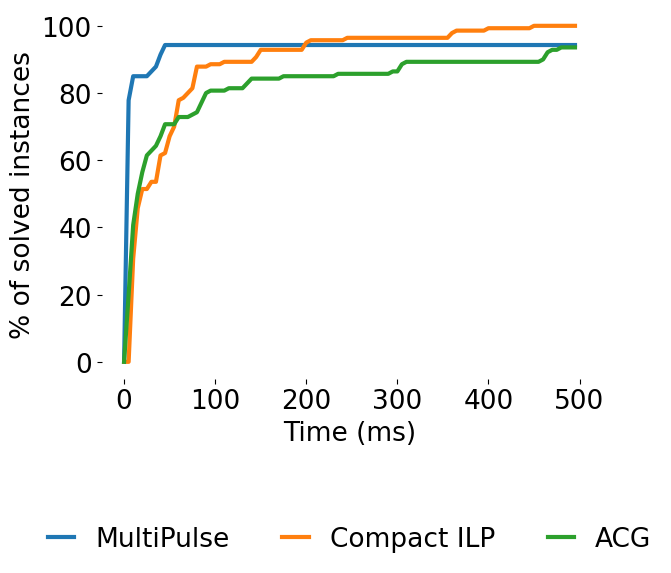} }}%
    \caption{Performance profiles for feasible instances from Grid and realistic topologies.}%
    \label{fig:performance}
\end{figure}

\subsubsection{Unfeasible instances}

We now compare the capabilities of the 3 algorithms to detect a non-trivial unfeasibility on the two kinds of topologies Grid and realistic.

\begin{figure}
    \centering
    \subfloat[\centering Unfeasible instances from Grid topologies.]{{\includegraphics[width=0.45\linewidth]{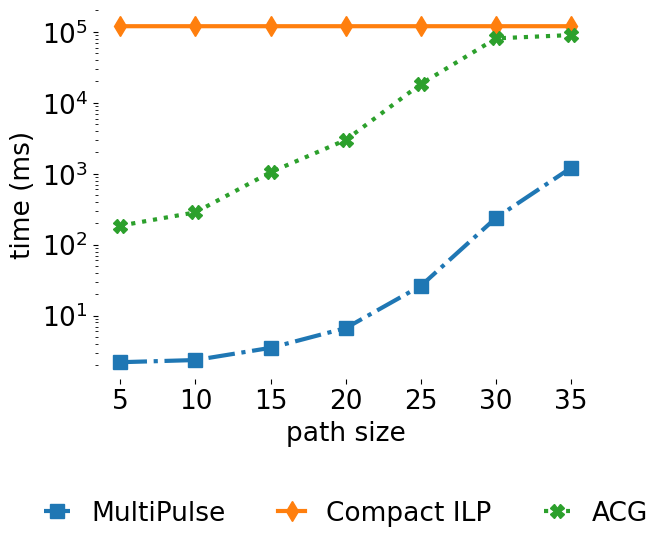} }}%
    \qquad
    \subfloat[\centering Unfeasible instances from realistic topologies.]{{\includegraphics[width=0.45\linewidth]{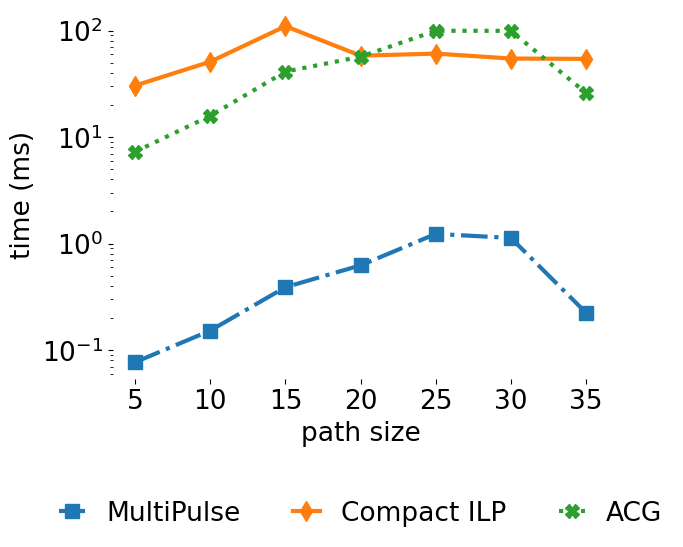} }}%
    \caption{Unfeasibility  detection  for Grid and realistic topologies.}%
    \label{fig:unfeasibility}
\end{figure}

On unfeasible instances from Grid topologies, see Fig. \ref{fig:unfeasibility}, we observe that compact ILP cannot detect unfeasibility. The time limit is reached on all instances. 
In realistic topologies, MultiPulse is by far the best for detecting unfeasibility, while compact ILP and our ACG scheme are quite similar in the amount of time they require. 
This is due to the size of realistic instances that is much smaller than Grid instances.

\section{Conclusion}
\label{sec:conclusion}

In this paper, we developped a novel framework, called Atomic Column Generation ($ACG$), that allows to easily combine efficient existing algorithms, referred to as \emph{atomic}, without the need to develop a dedicated algorithm for the whole problem. Our proposal is derived from Dantzig-Wolfe decomposition, allows converging to an optimal consensus solution from a pool of  atomic algorithms. We showed that
this decomposition improves the continuous relaxation and describe the associated Branch-and-Price algorithm. We considered a specific use case in telecommunications networks where several Path Computation Elements (PCE) are combined as atomic algorithms to route traffic. We demonstrated the efficiency of ACG on the resource-constrained shortest path problem associated with each PCE and showed that it remains competitive with benchmark algorithms.

In telecommunication context, future work includes multi-paths and multi-commodity flow extensions. Actually, we have already developed a functional multi-paths version but we preferred to focus on the single-path case to ease presentation.

\bibliographystyle{abbrv-networks}
\bibliography{biblio}

\end{document}